\documentclass[11pt]{article}
\usepackage{chao}
\usepackage{subcaption}
\usepackage{mathtools}
\usepackage{todonotes}
\usepackage{arydshln,leftidx,mathtools}
\usepackage{thm-restate}
\usepackage{physics}

\newcommand{\F}{\mathbb{F}}

\title{Solving systems of linear equations through zero forcing set and application to lights out}
\DeclareMathOperator\supp{supp}
\DeclareMathOperator\spn{span}

\author{Jianbo Wang\footnote{ \email{jbwang962@gmail.com}, University of Electronic Science and Technology of China, School of Computer Science and Engineering.} \and Chao Xu\footnote{ \email{cxu@uestc.edu.cn}, University of Electronic Science and Technology of China, School of Computer Science and Engineering.} \and Siyun Zhou\footnote{ \email{zhousiyun@std.uestc.edu.cn}, University of Electronic Science and Technology of China, School of Mathematical Sciences.}}
\date{\today}
\usepackage{xspace}

\begin{document}
\maketitle

\begin{abstract}
Let $\mathbb{F}$ be any field, we consider solving $Ax=b$ repeatedly for a matrix $A\in\mathbb{F}^{n\times n}$ of $m$ non-zero elements, and multiple different $b\in\mathbb{F}^{n}$. If we are given a zero forcing set of $A$ of size $k$, we can then build a data structure in $O(mk)$ time, such that each instance of $Ax=b$ can be solved in $O(k^2+m)$ time. As an application, we show how the lights out game in an $n\times n$ grid is solved in $O(n^3)$ time, and then improve the running time to $O(n^\omega\log n)$ by exploiting the repeated structure in grids.
\end{abstract}

\section{Introduction}

Solving the linear system $Ax=b$ has been a fundamental problem in mathematics. We consider the problem in its full generality. Given a field $\F$, a matrix $A\in \F^{n\times n}$ and a vector $b\in \F^n$, we are interested to find an $x\in \F^n$ such that $Ax = b$. Additionally, a zero forcing set associated with $A$ is also given.

As arguably the most important algorithmic problem in linear algebra, there are many algorithms designed for solving systems of linear equations, we refer the reader to \cite{Golub2013-ly}.
However, most works are for matrices over the real, complex or rational number fields. As this work is concerned with general fields, we will describe some known general algorithms below.
The Gaussian elimination is a classic method for solving systems of linear equations, which requires $O(n^3)$ time. If the system is of full rank, one can reduce the problem to a single matrix multiplication, by taking $x=A^{-1}b$. In general cases, the more recent LSP decomposition, which is a generalization of the LUP decomposition, allows the running time to match the matrix multiplication \cite{ibarraGeneralizationFastLUP1982,Jeann.2006}. If one views the non-zeros of the matrix as an adjacency matrix of a graph, the property of the graph can be then used to speed up the algorithm. For example, if the graph is planar, then the nested-dissection technique can obtain a $O(n^{3/2})$ time algorithm when the field is real or complex \cite{lipton_generalized_1979}. Later, it was generalized to any non-singular matrix over arbitrary field, and the running time was also improved to $O(n^{\omega/2})$ \cite{Alon.Y2013},  where $\omega<2.3728596$ is the matrix multiplication constant \cite{matrixmultiplication}. These algorithms are fairly complicated, but share one similarity with our work in taking the advantage of properties of the graph to improve the running time of the algorithm. 

Zero forcing sets were first studied in \cite{AIMM.2008} on graphs relating to the maximum nullity of a matrix, and the results were later expanded to directed graphs \cite{Bario.F.H.H.H.V.S.S.S.S.S2009}. The zero forcing set captures some "core" information of the linear system: $Ax=b$ is uniquely determined by the values of $x$ on a zero forcing set. Therefore this implies often that one just have to observe part of $x$ to recover the entirety of $x$. This idea leads to independent discovery of zero forcing set by physicists for control of quantum systems \cite{PhysRevLett.99.100501,Severini_2008}. Later it was shown to be applicable to Phasor Measurement Units (PMU) to monitor power networks \cite{6062919}.
Most studies of zero forcing set are on the algebraic and combinatorial properties. On the computational front, finding the smallest zero forcing set of an undirected graph is not known to be solvable in polynomial time, and it is likely NP-hard. The related problem of finding the minimum 0-1-weighted zero forcing set in an undirected graph is hard \cite{AazamiAshkan2008}. For directed graphs, finding the zero forcing set is NP-hard \cite{TREFOIS2015199}. Exact algorithms for finding the smallest zero forcing set was considered, see \cite{RePEc:eee:ejores:v:273:y:2019:i:3:p:889-903} for a survey.

This work was inspired by an algorithm for the lights out game. In the game, there is a light on each vertex of the graph, and each light is in a state either on ($1$) or off ($0$). There is also a button on each vertex. Pressing the button would flip the state of the light of itself and all its (out-)neighbors. The goal is to turn off all the lights. 
The lights out game is equivalent to solving a system of linear equations $Ax+b=0$ in $\F_2$, where $A$ is the adjacency matrix of the graph $G$, and $b$ is the state of the lights and $\F_2$ is the finite field of order $2$. The interpretation is that $x_v=1$ means pressing the button at vertex $v$, and $b_v$ is the initial state of the light at vertex $v$. Note in $\F_2$, $-b=b$, hence it is equivalent to solving $Ax=b$. There is a large amount of research in the lights out game, see \cite{Fleischer2013} for a survey. Finding a solution of the lights out game with the minimum number of button presses is NP-hard \cite{Berma.B.H2021}.

We focus on the case when $G$ is an $n\times n$ grid, which corresponds to an $n^2\times n^2$ matrix. Gaussian elimination would take $O(n^6)$ time. An alternative approach is the light-chasing algorithm \cite{Leach.2017}. Since the lights in the first row uniquely affect the states of the remaining rows, one can then look up which action on the first row should be taken according to the states of the last row. However, the literature does not provide the 
strategy for finding the lookup table. Wang claimed there is a $O(n^3)$ time algorithm that given the state of the last row, the state of the first row can be found through solving a linear equation of an $n$-square matrix \cite{wang}. Wang's result is the motivation behind this work. However, there is no proof of its correctness. If the adjacency matrix has full rank, then by nested dissection, there is a randomized Las Vegas algorithm with running time of $O(n^\omega)$ \cite{Alon.Y2013}.  Unfortunately, there are infinite many singular adjacency matrices \cite{HUNZIKER2004465}, the list of such matrices can be found in OEIS A117870 \cite{oeisA117870}. A SAGE package  for computing the solution of the $n\times n$ grid lights out game when all the lights are on is provided in \cite{scanlightsout}. 

\paragraph{Our contribution}
We generalize Wang's method for the lights out game to solving an arbitrary system of linear equations $Ax=b$, and give a formal proof of its correctness. We define a structure, the core matrix $B$ of $A$, such that one can solve the system of linear equations over $B$ instead of $A$, and then lift it to a solution of $A$ in linear time. As a consequence, we obtain the following algorithmic result.

\begin{restatable}[Main]{theorem}{main}
\label{thm:main}
Given a matrix $A\in\mathbb{F}^{n\times n}$ of $m$ non-zero elements, and a zero forcing set of $A$ of size $k$ where $k\leq n$. A data structure of size $O(k^2)$ can be computed in $O(mk)$ time, such that for each $b\in\mathbb{F}^{n}$, solving $Ax=b$ takes $O(k^2+m)$ time.
\end{restatable}

We also show that the lights out game on an $n\times n$ grid can be solved in $O(n^{\omega}\log n)$ time by finding the core matrix using an alternative method.
Additionally, we prove some linear algebraic properties of the zero forcing set, which might be of independent interest.

\section{Preliminaries}

Define an index set $[n]=\{1,\ldots,n\}$. We consider an algebraic model, where every field operation takes $O(1)$ time, and each element in the field can be stored in $O(1)$ space.

It is useful to think about a vector $x$ indexed by elements $X$ is a function $x: X\to \F$. Hence, we write $x\in \F^X$. Define $\supp(x)$ to be the set of non-zero coordinates of $x$.

For a matrix $A$ and a set of row indices $R$ and a set of column indices $C$, we define $A_{R,C}$ to be the submatrix of elements indexed by the row and columns in $R$ and $C$, respectively. We use $*$ to mean the set of all row indices or column indices, depending on context. For example, $A_{R,*}$ to denote the submatrix of $A$, which is composed of the rows indexed by $R$. We define $A_{i,*}$ as the $i$th row vector of $A$. Similarly, $A_{*,j}$ is the $j$th column vector. $A_{i,j}$ is the element in the $i$th row and $j$th column. For a vector $x$, $x_i$ is the scalar at index $i$, and $x_I$ is the subvector of the elements in the index set $I$. 

Solving $Ax=b$ is finding a vector $x$ that satisfies $Ax=b$. For a directed graph $G=(V,E)$ on $n$ vertices, and a field $\F$, we define $\mathcal{M}(G,\F)$ to be the set of all matrices $A\in \F^{n\times n}$ such that for all $u\neq v$, $A_{u,v}\neq 0$ if and only if $(u,v)\in E$. We do not have any restriction on $A_{v,v}$. Let $N^+(u)=\{v \mid (u,v) \in E \}$ be the set of all out-neighbors of $u$.

We briefly review some basic concepts and results related to the zero forcing. Consider the process of coloring a graph. We start with a set of blue colored vertices, denoted by $Z$. If there exists a blue vertex $v$ with exactly one non-blue out-neighbor, say $u\in V$, then color $u$ blue. The operation of turning $u$ blue is called \emph{forcing}, and we say $u$ is forced by $v$. If this process ends up with a situation where all vertices are colored blue, then we say the initial set $Z$ is a \emph{zero forcing set}. The \emph{zero forcing number} $Z(G)$ of $G$ is the size of the smallest zero forcing set. We also call $Z$ a zero forcing set of $A$ if $A\in \mathcal{M}(G,\F)$ for some field $\F$. The following proposition gives the reason for the name zero forcing.

\begin{proposition}[\cite{AIMM.2008,Hogbe.2010}]
For a zero forcing set $Z$ of $A$, if $x\in \ker(A)$, then $x_Z=0$ implies $x=0$. Namely, $x$ vanishing at zero forcing set forces $x$ to be $0$.
\end{proposition}

The converse is not true in general. Indeed, we can simply take $A$ to be a $2\times 2$ identity matrix. Then, $Z=\{1,2\}$ is the unique zero forcing set of $A$, and $\ker(A)=\{0\}$. If $x\in \ker(A)$ and $x_1=0$, then we have $x=0$. But $\{1\}$ is clearly not a zero forcing set.

The order of forces for a zero forcing set $Z$ is not unique, although the final coloring is \cite{AIMM.2008,Hogbe.2010}. For simplicity, we avoid the issue by considering a particular chronological list of forces $\pi$, which picks the smallest indexed vertex that can be forced. $\pi$ is a total ordering of the vertices such that the elements in $Z$ are ordered arbitrarily, and smaller than all elements in $V\setminus Z$. For each $v,u\in V\setminus Z$, $v\leq_\pi u$ if $v$ is forced no later than $u$. The forcing graph is a graph where there is an edge between $u, v$ if $u$ forces $v$. It is well known that such graph is a set of node disjoint paths that begin in $Z$ \cite{Hogbe.2010}. Hence, if $v$ forces $u$, we can define $u^{\uparrow}=v$ to be the \emph{forcing parent}, and correspondingly, $u$ is called the \emph{forcing child} of $u^{\uparrow}$. A vertex is a \emph{terminal}, if it does not have a forcing child. Let $T$ be the set of terminals, then $|T|=|Z|$. 

Consider the algorithm $\textsc{Forcing}(A,Z,b)$ as summarized in \autoref{alg:forcing}, which takes a matrix $A$, a vector $b$, and a zero forcing set $Z$ of $A$ as inputs, and updates $x_u$ for all $u\in V$ in each round. We set $x_u=0$ as initialization. 
Then, for each $u\in V\setminus Z$, we update $x_u$ according to $\left(b_{u^\uparrow} - \sum_{v\in N^+(u^{\uparrow})\setminus \set{u}} A_{u^{\uparrow},v} x_v\right)/ A_{u^{\uparrow},u}$, iteratively in the forcing order. Note that it is equivalent to setting $x_u \gets \left(b_{u^\uparrow} - A_{u^{\uparrow},*} x\right)/A_{u^\uparrow,u}$, since $x_u$ is previously $0$. 

    \begin{figure}[ht]
    \centering
    \begin{algorithm}
      \textsc{Forcing}$(A,Z,b)$\+
      \\  $x\gets 0$
      \\  for $u\in V\setminus Z$ ordered by some forcing sequence $\pi$\+
      \\    $x_u \gets  \left(b_{u^\uparrow} - A_{u^{\uparrow},*} x\right)/ A_{u^{\uparrow},u} $\-
      \\  return $x$
    \end{algorithm}
    \caption{The forcing operation.}
    \label{alg:forcing}
    \end{figure}

The following result formalizes the relation between the solution to a given linear system and the corresponding zero forcing set. 
\begin{proposition}\label{pro:forcing}
After the value of $x_u$ is updated in $\textsc{Forcing}(A,Z,b)$, we have $A_{u^\uparrow,*}x=b_{u^\uparrow}$ for all $v\leq_\pi u$. In particular, if $Ax=b$ has a solution where $x_Z=0$, then $\textsc{Forcing}(A,Z,b)$ finds such solution.
\end{proposition}
\begin{proof}
One can prove it by induction. 
The trivial case is when no $x_u$ has been updated, then the conclusion is vacuously true. If some $x_u$ is updated by $\textsc{Forcing}(A,Z,b)$, we let $x'$ be the vector before the update, and $x''$ be the one after the update. Then we get $x'_v=x''_v$ for all $v\neq u$, and $x'_u=0$. And $A_{u,*}x'' = A_{u,*}x' + A_{u^\uparrow,u}x''_u = A_{u,*}x' + A_{u^\uparrow,u}(b_{u^\uparrow} - A_{u^\uparrow,*}x')/A_{u^{\uparrow},u} = b_{u^\uparrow}$. Also, for all $v <_\pi u$, $A_{v^\uparrow,*}x''=A_{v^\uparrow,*}x'=b_{v^\uparrow}$.
\end{proof}

Observe that the running time of $\textsc{Forcing}(A,Z,b)$ is $O(m)$, where $m$ is the number of non-zero elements in $A$.
\begin{theorem}\label{thm:forcingprop}
Let $Z$ be a zero forcing set of $A$. The following statements are true.
 \begin{enumerate}
   \item The columns in $V\setminus Z$ are linearly independent.
   \item If $Ax=b$, $Ay=b$, and $x_Z = y_Z$, then $x=y$.
   \item Given $x'\in F^V$ such that $\supp(x')\subseteq Z$. If $Ax=b$ for some $x$ such that $x_Z=x'_Z$, then $x=\textsc{Forcing}(A,Z,b-Ax')+x'$.
 \end{enumerate}
\end{theorem}
\begin{proof}
See \cite{kenter_error_2019} for proof of the first two statements. Note in \cite{kenter_error_2019}, for the second statement, only the $b=0$ version was proven, but the same proof still works for general $b$.
For the third statement, if $Ax=b$ where $x_Z = x'_Z$, then $A(x'+(x-x')) = b$, or in other words, $A(x-x')=b-Ax'$. Hence by \autoref{pro:forcing}, $\textsc{Forcing}(A,\pi,b-Ax')=x-x'$, and therefore $(x-x')+x' = x$.
\end{proof}

The \emph{LSP decomposition} of $A\in \F^{m\times n}$ takes the form $A=LSP$, where $L\in \F^{m\times m}$ is a lower triangular matrix with value $1$ in the diagonals, $S\in \F^{m\times n}$ can be reduced to a upper triangular matrix if all zero rows are deleted and elements in the main diagonal after the zero row deletion is non-zero, and $P\in \F^{n\times n}$ is a permutation matrix. The LSP decomposition can be found in $O(mnr^{\omega-2})$ time \cite{Jeann.2006}, where $r$ is the rank of $A$. 

\begin{theorem}\label{thm:lsp}
  One can build a data structure on an $n$-square matrix $A$ in $O(n^\omega)$ time, such that $Ax=b$ can be solved in $O(n^2)$ time for each $b$. 
\end{theorem}
\begin{proof}
  Computing the LSP decomposition of $A\in \F^{n\times n}$ takes $O(n^\omega)$ time. Given the LSP decomposition, solving $Ax=b$ can be obtained through solving $Ly=b$ and $SPx=y$ through back substitution  \cite{ibarraGeneralizationFastLUP1982}, which takes $O(n^2)$ time. We just have to store all the involved matrices, which takes $O(n^2)$ space.
  \end{proof}

%%%%%%%
\section{Solving $Ax=b$ through zero forcing set}
In this section, we first introduce the \emph{core matrix} that serves as the key part of our algorithm, followed by its theoretical guarantees. Based on the core matrix, we then present the detailed algorithm for solving $Ax=b$ with a given zero forcing set of size $k$, as well as the corresponding computational analysis. 

To greatly simplify the exposition, we set up the instance we are working with. Throughout this section, we fix a directed graph $G=(V,E)$, a field $\F$, a matrix $A\in \mathcal{M}(G,\F)$, a zero forcing set $Z$ of $G$, a forcing order $\pi$, and the terminals $T$ under the forcing order. 
Since the operations are performed on indices, without loss of generality, we assume that $V=[n]$ and $Z=[k]$.

Let $L(x)=\textsc{Forcing}(A,\pi,x)$, $R(x) = x-AL(x)$. As we will see, both are linear functions. It can be also observed that $\supp(R(b))\subset T$ and $\supp(L(b))\subset V\setminus Z$ for all $b$. $|V|=n$, $|Z|=k$, and the number of non-zero elements in $A$ is $m$.

\subsection{Core matrix}
As \autoref{thm:forcingprop} suggests, the solution to $Ax=b$ can be obtained by knowing only $x_Z$. Hence, it is natural to think of finding the correct $x_Z$ instead of acquiring the full information of $x$ for solving $Ax=b$.

Let $a_v=A_{*,v}$ be the column of $A$ indexed by $v$. Define a $k\times k$ matrix $B\in \F^{T\times Z}$ as $B=(RA)_{T,Z}$. In other words, the $v$th column equals $R(a_v)_T$. The matrix $B$ is called the \emph{core matrix} of $A$.
We will show that for any $b$, $Ax=b$ if and only if $Bx_Z=R(b)_T$. Hence, solving the equation $By=R(b)_T$ along with some post processing is sufficient to give a solution of $Ax=b$.

First, we prove the linearity of functions $L$ and $R$. 
\begin{lemma}\label{lem:linear}
$L$ and $R$ are linear functions.
\end{lemma}
\begin{proof}
The equation $Rb = (I-AL)b$ indicates that if $L$ is linear, $R$ must be linear. Hence we only need to show the linearity of $L$.

The proof can be derived by encoding each forcing operation using matrices. Define $D(u)\in \F^{n\times n}$ for each $u$ with $D(u)_{u,*} = A_{u^\uparrow,*}/A_{u^\uparrow,u}$, $D_{v,v}=1$ if $v\neq u$, and $0$ everywhere else. Also, define $E(u)\in \F^{n\times n}$ to be the matrix that is $0$ everywhere except $E(u)_{u^\uparrow,u}=1/A_{u^\uparrow,u}$.
We can then define $M(v) \in \F^{(2n)\times (2n)}$ for each $v\in V$, which is written in a block form as follows
\[
\begin{bmatrix}
-D(u) & E(u)_{u,u^\uparrow}\\
0_n & I_n
\end{bmatrix},
\]
where $I_n$ and $0_n$ are the identity and the zero matrix of order $n$, respectively.

Let $s=(x_1,\ldots,x_n,b_1,\ldots,b_n)^T$. Observe that $M(v)s$ and $s$ differ only in one coordinate, $x_u$, and the value of $x_u$ is $b_{u^\uparrow}/A_{u^\uparrow,u}-A_{u^\uparrow,*}/A_{u^\uparrow,u}x$, precisely the relation in forcing.
 
Now, consider $\hat{s} = (0,\ldots,0,b_1,\ldots,b_n)^T\in\F^{2n}$, which can be obtained by a linear transformation from $b$. 
Let $\hat{s}' = M(v_n)\ldots M(v_{k+1}) \hat{s}$, where $v_{k+1},\ldots,v_n$ are the vertices in $V\setminus Z$ ordered by the forcing order. Hence, the vector $\hat{s}'$ can also be obtained by a linear transformation from $b$. 
Assigning the values of the first $n$ coordinates of $\hat{s}'$ to $x'$, $x'$ then coincides with the value of $L(b)$. This implies that $x'$ is obtained from $b$ through a linear transformation, and thus, $L$ is linear. 
\end{proof}

Based on the linearity above, we will abuse the notation and let $L$ and $R$ to represent the matrices for the corresponding linear transformations.

\begin{lemma}\label{lem:cv}
$R a_v=0$ for $v\not \in Z$.
\end{lemma}
\begin{proof}
The forcing algorithm given in \autoref{alg:forcing} shows that if $v\neq u$, we will then have $x_v=1$, and $x_u=0$. Note that $Ax=a_v$, and therefore $R a_v = a_v - Ax = 0$.
\end{proof}

Before delving into the key theorem, we introduce some useful notations and definitions that will be repeatedly used in the remaining part of this section.
Let $\rank(A)=r\geq n-k$. Define $A' = A_{*,[k]}\in  \F^{n\times k}$, and $A''=A_{*,[n]\setminus [k]} \in \F^{n\times (n-k)}$.
To facilitate the analysis, we
further decompose $A'$ into a $1\times 2$ block matrix form $A'=[A'_1\quad A'_2]$ such that $A'_1 = [a_1\quad \ldots\quad a_p]$, $A'_2 = [a_{p+1}\quad \ldots \quad a_k]$, $\rank([A'_1 \quad A''])=r$, and $\rank(A'_1) = p=r-(n-k)$. The matrix $A$ can be then rewritten as
\begin{equation}
\label{eq:A}
A = [A'_1\quad A'_2\quad A''].
\end{equation}

The following result presents the rank-preserving property of $A'_1$ under the linear mapping $R$. Based on this property, together with the assumption of the existence of the solution to $Ax=b$, we can then arrive at the main theoretical result of this work, which will be detailed in \autoref{thm:epic} and \autoref{thm:mainprop}.

\begin{lemma}\label{lem:rank}
$\rank(R A'_1) = \rank(A'_1)$.
\end{lemma}
\begin{proof}
For any $s\in\mathbb{F}^n$, there exist $\gamma_{k+1},\cdots,\gamma_{n}$, such that
\[
s+\sum_{i=k+1}^n\gamma_{i}a_i=Rs,
\]
which implies that $(R-I)s\in\spn(A'')$. Define another linear mapping $Q$ as $Q = R-I$, and we have $\spn Q\subset \spn A''$. Then, we can write $RA_1'$ as
\[
RA_1' = (Q+I)A_1'=QA_1'+A_1',
\]
where $QA_1'\in \spn Q\subset \spn A''$. 
Hence, the columns of $QA_1'$ and the ones of $A_1'$ are linearly independent, which immediately gives $\rank(R A'_1) = \rank(A'_1)$.
\end{proof}

The following theorem plays a pivotal role in the theoretical guarantees for our algorithm. 
\begin{theorem}\label{thm:epic}
Given $A\in \F^{n\times n}$ of the form (\ref{eq:A}).
Let $M\in \F^{k\times n}$ ($k\leq n$) such that $A''\in \ker(M)$, and $\rank(MA'_1)=\rank(A'_1)$.
For $b\in \F^n$, suppose that $Ax=b$ has a solution.
If $MA'y=M b$ for some $y\in \F^k$, then there exists $x\in \F^n$ such that $Ax=b$ and $x_Z=y$.
\end{theorem}
\begin{proof}
Since $A_2'\in\spn(A_1', A'')$, we can rewrite $A_2'$ as
$A_2' = A_1'C_1' + A''C''$ where $C_1'\in\mathbb{F}^{p\times(k-p)}$ and $C''\in\mathbb{F}^{(n-k)\times(k-p)}$ are coefficient matrices. 
The existence of the solution to $Ax=b$ indicates that $b\in\spn(A_1', A'')$. Similarly, we can write $b = A_1'C_b'+A''C_b''$ where $C_b'\in\mathbb{F}^{p}$ and $C''_b\in\mathbb{F}^{n-k}$ are coefficient vectors. Then, we have
\begin{align*}
  M[A_1'\quad A_2'] y &= M[A_1'\quad A_1'C_1' + A''C''] y= M[A_1'\quad  A_1'C_1']y,\\
  Mb &= M(A'_1C_b'+A''C_b'')=MA_1'C_b',
\end{align*}
where the last equalities of the above two formulas follow from $A''\in \ker(M)$.
Then, the equation $MA'y = Mb$ can be written as
$[MA_1'\quad MA_1'C_1'] y = MA_1'C_b'$.
Decomposing $y\in\mathbb{F}^k$ into $y=\begin{bmatrix}y_1 \\y_2\end{bmatrix}$ with $y_1\in \F^p, y_2\in \F^{k-p}$ further
leads to a new homogeneous linear system
\[
MA_1'(y_1+C_1' y_2 - C_b')=0.
\]
The condition $\rank(MA'_1)=\rank(A'_1)$ gives $\rank(MA'_1)=p$, or to say, $MA'_1$ is of full rank. Thus, we have $y_1+C_1'y_2 - C_b' = 0$.

In a similar manner, we decompose $x\in\mathbb{F}^n$ into three subvectors $x_1\in \F^p$,$x_2\in \F^{k-p}$, and $x_3\in \F^{n-k}$. Then $Ax=b$ gives
\[
[A_1'\quad A_1'C_1' + A''C''\quad A''] \begin{bmatrix}x_1 \\x_2\\x_3\end{bmatrix} = A_1'C_b' + A''C_b''.\]
Taking $x_1=y_1$, $x_2=y_2$ yields
\[
A_1'y_1+A_1'C_1'y_2 + A''C''y_2 + A''x_3 = A_1'C_b' + A''C_b'',
\]
which can be rearranged as
\[
A_1'(y_1 + C_1'y_2-C_b') + A''(C''y_2+x_3-C_b'') = 0.
\]
Since $C''y_2+x_3-C_b''=0$, we get
\[
A''(C''y_2+x_3-C_b'') = 0.
\]
Now, we can set $x_3 = C_b'' - C''y_2$, which immediately gives a solution of $Ax=b$ of the form 
\[
x = \begin{bmatrix}y_1 \\y_2\\C_b'' - C''y_2\end{bmatrix}.
\]
\end{proof}

\begin{theorem}\label{thm:mainprop}
Let $B$ be the core matrix of $A$, and $Ax=b$ has at least one solution. If $By=R_{T,*}b$, then there exists $x$ such that $Ax=b$ and $x_Z=y$. 
\end{theorem}
\begin{proof}
Since $\supp(Rx),\supp(R_{T,*}x)\in T$ for all $x$, we have $\rank(RA'_1)=\rank(R_{T,*}A'_1)$. By \autoref{thm:epic} with $M=R_{T,*}$ and \autoref{lem:rank}, we conclude that such solution always exists.
\end{proof}

%%%%%
\subsection{The algorithm}
We first provide the algorithm \textsc{FindCore}, see \autoref{alg:findcore}, for computing the
core matrix effectively, with the computational cost given in \autoref{thm:findcore}.
   \begin{figure}[ht]
    \centering
    \begin{algorithm}
      \textsc{FindCore}$(A,Z,b)$\+
      \\  $\pi, T\gets$ the forcing ordering and terminal set
      \\  for $v\in Z$\+
      \\    $z \gets$ \textsc{Forcing}$(A,Z,A_{*,v})$
      \\    $B_{*,v} \gets (A_{*,v})_T-A_{T,*}z$\-
      \\  Compute the LSP decomposition of $B$
      \\  return $B$
    \end{algorithm}
    \caption{Find the core matrix.}
    \label{alg:findcore}
    \end{figure}
    
\begin{theorem}\label{thm:findcore}
The algorithm \textsc{FindCore} takes $O(mk)$ time and $O(k^2)$ space.
\end{theorem}
\begin{proof}
Computing $L(A_{*,v})$ for all $v\in Z$ takes $O(m(1+|Z|))=O(mk)$ time. The computation of $B_{*,v}$ for a $v\in Z$ consists of a vector-vector difference and a matrix-vector product, which can be implemented in $O(m)$ linear time. Thus, computing the core matrix takes $O(mk)$ time in total. 
\end{proof}

Next, the algorithm \textsc{SolveLinearSystemGivenCore}$(A,Z,B,b)$ in \autoref{alg:solver2} shows how the solution of $Ax=b$ is obtained using the computed core matrix and all the information about the zero forcing set. The computational cost is provided in \autoref{thm:solve}.

    \begin{figure}[ht]
    \centering
    \begin{algorithm}
      \textsc{SolveLinearSystemGivenCore}$(A,Z,B,b)$\+
      \\  $\pi, T\gets$ forcing sequence and terminal set $T$
      \\  $z\gets \textsc{Forcing}(A,Z,b)$
      \\  $b'\gets b_T - A_{T,*}z$
      \\  $y \gets$ solution to $By = b'$
      \\  if $y$ does not exists\+
      \\    return \texttt{"NO SOLUTION"}\-
      \\  $x'\gets $ the vector where $x'_Z=y$ and 0 everywhere else
      \\  $x\gets x'+\textsc{Forcing}(A,Z,b-Ax')$
      \\  if $Ax=b$\+
      \\    return $x$\-
      \\  else\+
      \\    return \texttt{"NO SOLUTION"}\-
      \end{algorithm}
    \caption{Solve a linear system $Ax=b$ given a core matrix.}
    \label{alg:solver2}
    \end{figure}

By \autoref{thm:mainprop}, we know there exists a solution that matches $y$ at the zero forcing set. By \autoref{thm:forcingprop}, we obtain the remaining part of the solution through forcing.

\begin{theorem}\label{thm:solve}
Given a matrix $A\in\mathbb{F}^{n\times n}$ of $m$ non-zero elements, its zero forcing set of size $k$, and a core matrix $B$ represented by its LSP decomposition. The system of linear equations $Ax=b$ can be solved in $O(k^2+m)$ time.
\end{theorem}
\begin{proof}
Following the algorithm in \autoref{alg:solver2},
the computation of $b'= R_{T,*}b$ by forcing takes $O(m)$ time. By \autoref{thm:lsp}, if the LSP decomposition of the matrix is given, then solving a system of linear equations on $k$ variables and $k$ equations takes $O(k^2)$ time. Once the solution of the linear system for the core matrix is obtained, we can find the solution to the original problem through forcing in $O(m)$ time. The total running time is then $O(k^2+m)$. 
\end{proof}

Combining \autoref{thm:findcore} and \autoref{thm:solve}, we finally arrive at our main theorem.
\main*

%%%%%%%
\section{Lights out game on a grid}
In this section, we show that the lights out game in an $n\times n$ grid can be solved in $O(n^\omega\log n)$ time. Consider the lights out game on an $n\times n$ grid graph. We number the vertices in position $(i,j)$ with index $(i-1)n+j$. Hence, the vertices are in $[n^2]$. 
Let $Z=[n]$, we see that $Z$ is a zero forcing set. If we apply \autoref{thm:main} directly to the adjacency matrix of the grid graph, then in this special case, we will obtain precisely the algorithm of \cite{wang}. Since $m=n^2$ and $k=n$, the algorithm takes $O(n^3)$ time. The computational bottleneck is the calculation of the core matrix. Fortunately, by exploiting the repeated structure in grids, we can significantly improve the running time of computing the core matrix to $O(n^\omega\log n)$.

The forcing operation for the lights out game is greatly simplified, because the field is $\F_2$. In this case, $x_u=1$ or $x_u=0$ can be interpreted as pressing the button at vertex $u$ or not, respectively. The $b_u$ can be understood as the state of light at vertex $u$, where the value $1$ means on, and  $0$ means off.
When operating on the $u$th vertex, the forcing operation sets $x_u=1$ if and only if $b'_{u^{\uparrow}}=1$. Here $b'$ is the state of the board after applying all previous button presses. In other words, the forcing operation is iteratively setting $x_u=b'_{u^{\uparrow}}$.

We want to encode the operation of forcing, where we are given the state of first and second rows and want to compute the force operations on the second row ensuring the states of the first row can be all $0$. It should then output the state of the second row and the state of the third row after all the button presses.
To this end, we define such matrix to be $N(n)\in \F_2^{2n\times 2n}$. 
The vertices of first row are indexed from $1$ to $n$. The vertex of the second row below the vertex $i$ is $n+i$. One can easily verify that $N$ could be written in the following block form 
\[
N(n) = \begin{bmatrix}
N'(n) & I_n\\
I_n & 0_n
\end{bmatrix},
\]
where $N'(n)$ is the matrix satisfying that $N'(n)_{i,j}=1$ if and only if $|i-j|\leq 1$.

\begin{figure}
\centering
\[ \left[ \begin{array}{c:c} \begin{array}{c c c} 
1 & 1 & 0\\
1 & 1 & 1\\
0 & 1 & 1\\
\end{array} &  
\begin{array}{c c c}
1& 0& 0\\
0& 1& 0\\
0& 0& 1\\ 
\end{array} \\ \hdashline
\begin{array}{c c c}
1& 0& 0\\
0& 1& 0\\
0& 0& 1\\ 
\end{array} & \begin{array}{c c c}
0& 0& 0\\
0& 0& 0\\
0& 0& 0\\ 
\end{array} \end{array} \right] \]
    \caption{The matrix $N(3)$. }
    \label{fig:theN}
\end{figure}

Now, define $M=N^{n-1}(n)$, which requires $O(n^\omega\log n)$ time to compute using exponentiation by squaring. Let $a_j$ to be the $j$th column of $A$. Since $\supp(a_j)$ is a subset of the first two rows if $j\in [n]$, we first let $a_j' = (a_j)_{[2n]}$. Next, we compute $y=Ma_j'$. Let $t$ be the last $n$ coordinates of $y$. This then gives us the desired $R_{T,*}a_j$. Note that we can batch all the multiplications together, that is, we can compute $M[a_j'|j\in [n]] = M A_{[n],[2n]}$ in one go, which takes $O(n^\omega)$ running time, and recover the desired core matrix from it. The procedure described above is summarized in \autoref{alg:solverbgrid}. 

    \begin{figure}[ht]
    \centering
    \begin{algorithm}
      \textsc{FindGridCore}$(n)$\+
      \\  $A\gets $ adjacency matrix of an $n\times n$ grid
      \\  $M\gets (N(n))^{n-1}$
      \\  return $(M A_{[n],[2n]})_{[n],[n]}$
    \end{algorithm}
    \caption{Find core matrix for a grid graph.}
    \label{alg:solverbgrid}
    \end{figure}

At this point, we have showed a $O(n^{\omega}\log n)$ time algorithm to solve the lights out problem on an $n\times n$ grid.
Moreover, the algorithm can also be applied to $n\times m$ grids with $n\leq m$, and the running time becomes $O(n^\omega \log m + nm)$ accordingly.

\bibliographystyle{plain}
\bibliography{lights_out}

\begin{thebibliography}{10}

\bibitem{AazamiAshkan2008}
{Aazami, Ashkan}.
\newblock {\em Hardness results and approximation algorithms for some problems
  on graphs}.
\newblock PhD thesis, 2008.

\bibitem{AIMM.2008}
{AIM Minimum Rank \textendash{} Special Graphs Work Group}.
\newblock Zero forcing sets and the minimum rank of graphs.
\newblock {\em Linear Algebra and its Applications}, 428(7):1628--1648, April
  2008.

\bibitem{matrixmultiplication}
Josh Alman and Virginia~Vassilevska Williams.
\newblock {\em A Refined Laser Method and Faster Matrix Multiplication}, pages
  522--539.

\bibitem{Alon.Y2013}
Noga Alon and Raphael Yuster.
\newblock Matrix sparsification and nested dissection over arbitrary fields.
\newblock {\em Journal of the ACM}, 60(4):1--18, August 2013.

\bibitem{Bario.F.H.H.H.V.S.S.S.S.S2009}
Francesco Barioli, Shaun Fallat, H.~Hall, Daniel Hershkowitz, Leslie Hogben,
  Hein {Van der Holst}, Bryan Shader, Bryan Shader, Bryan Shader, Bryan Shader,
  and Bryan Shader.
\newblock On the minimum rank of not necessarily symmetric matrices: {{A}}
  preliminary study.
\newblock {\em The Electronic Journal of Linear Algebra}, 18, January 2009.

\bibitem{Berma.B.H2021}
Abraham Berman, Franziska Borer, and Norbert Hungerb{\"u}hler.
\newblock Lights {{Out}} on graphs.
\newblock {\em Mathematische Semesterberichte}, 68(2):237--255, October 2021.

\bibitem{RePEc:eee:ejores:v:273:y:2019:i:3:p:889-903}
Boris Brimkov, Caleb~C. Fast, and Illya~V. Hicks.
\newblock {Computational approaches for zero forcing and related problems}.
\newblock {\em European Journal of Operational Research}, 273(3):889--903,
  2019.

\bibitem{PhysRevLett.99.100501}
Daniel Burgarth and Vittorio Giovannetti.
\newblock Full control by locally induced relaxation.
\newblock {\em Phys. Rev. Lett.}, 99:100501, Sep 2007.

\bibitem{6062919}
Nathaniel Dean, Alexandra Ilic, Ignacio Ramirez, Jian Shen, and Kevin Tian.
\newblock On the power dominating sets of hypercubes.
\newblock In {\em 2011 14th IEEE International Conference on Computational
  Science and Engineering}, pages 488--491, 2011.

\bibitem{Fleischer2013}
Rudolf Fleischer and Jiajin Yu.
\newblock {\em A Survey of the Game ``Lights Out!''}, pages 176--198.
\newblock Springer Berlin Heidelberg, Berlin, Heidelberg, 2013.

\bibitem{scanlightsout}
Jim Fowler.
\newblock Scanning algorithms for \emph{Lights Out}.
\newblock \url{https://github.com/kisonecat/lights-out}.
\newblock Accessed: 2022-07-28.

\bibitem{Golub2013-ly}
Gene~H Golub and Charles~F Van~Loan.
\newblock {\em Matrix Computations}.
\newblock Johns Hopkins Studies in the Mathematical Sciences. Johns Hopkins
  University Press, Baltimore, MD, 4 edition, February 2013.

\bibitem{Hogbe.2010}
Leslie Hogben.
\newblock Minimum rank problems.
\newblock {\em Linear Algebra and its Applications}, 432(8):1961--1974, April
  2010.

\bibitem{HUNZIKER2004465}
Markus Hunziker, António Machiavelo, and Jihun Park.
\newblock Chebyshev polynomials over finite fields and reversibility of
  $\sigma$-automata on square grids.
\newblock {\em Theoretical Computer Science}, 320(2):465--483, 2004.

\bibitem{ibarraGeneralizationFastLUP1982}
Oscar~H. Ibarra, Shlomo Moran, and Roger Hui.
\newblock A generalization of the fast {{LUP}} matrix decomposition algorithm
  and applications.
\newblock {\em Journal of Algorithms}, 3(1):45--56, 1982.

\bibitem{oeisA117870}
OEIS~Foundation Inc.
\newblock Entry \href{http://oeis.org/A117870}{A117870} in the on-line
  encyclopedia of integer sequences.
\newblock \url{http://oeis.org/A117870}, 2022.

\bibitem{Jeann.2006}
Claude-Pierre Jeannerod.
\newblock {{LSP}} matrix decomposition revisited.
\newblock Technical Report 2006-28, {\'Ecole Normale Sup\'erieure de Lyon},
  September 2006.

\bibitem{kenter_error_2019}
Franklin~H.J. Kenter and Jephian C.-H. Lin.
\newblock On the error of a priori sampling: {Zero} forcing sets and
  propagation time.
\newblock {\em Linear Algebra and its Applications}, 576:124--141, September
  2019.

\bibitem{Leach.2017}
C.~David Leach.
\newblock Chasing the {{Lights}} in {{Lights Out}}.
\newblock {\em Mathematics Magazine}, 90(2):126--133, April 2017.

\bibitem{lipton_generalized_1979}
Richard~J. Lipton, Donald~J. Rose, and Robert~Endre Tarjan.
\newblock Generalized {Nested} {Dissection}.
\newblock {\em SIAM Journal on Numerical Analysis}, 16(2):346--358, April 1979.

\bibitem{Severini_2008}
Simone Severini.
\newblock Nondiscriminatory propagation on trees.
\newblock {\em Journal of Physics A: Mathematical and Theoretical},
  41(48):482002, oct 2008.

\bibitem{TREFOIS2015199}
Maguy Trefois and Jean-Charles Delvenne.
\newblock Zero forcing number, constrained matchings and strong structural
  controllability.
\newblock {\em Linear Algebra and its Applications}, 484:199--218, 2015.

\bibitem{wang}
Zheng Wang~(axpokl).
\newblock \zh{点灯游戏Flip game的$O(n^3)$算法}.
\newblock \url{https://zhuanlan.zhihu.com/p/53646257}, 2018.

\end{thebibliography}
\end{document}